\documentclass{IEEEtran}
\usepackage{cite}
\usepackage{amsmath,amssymb,amsfonts}
\usepackage{algorithmic}
\usepackage{textcomp}
\def\BibTeX{{\rm B\kern-.05em{\sc i\kern-.025em b}\kern-.08em
    T\kern-.1667em\lower.7ex\hbox{E}\kern-.125emX}}

\usepackage{graphicx}
\usepackage{epstopdf}
\usepackage{booktabs}
\usepackage{balance}
\usepackage{pbox}
\usepackage{psfrag}
\usepackage{tikz}
\usepackage{tikz-qtree}
\usetikzlibrary{trees}
\usetikzlibrary{fit,calc,arrows,positioning}
\usepackage{pgfplots}
\pgfplotsset{compat=newest}
\pgfplotsset{plot coordinates/math parser=false}
\newlength\figureheight
\newlength\figurewidth
\usepackage{graphicx}
\usepackage{subcaption}
\usepackage{mathtools}

\newcommand{\sat}{\operatorname{sat}}

\newcommand{\sgn}{\operatorname{sign}}

\newcommand{\Diag}{\operatorname{diag}}

\newcommand{\0}{{0}}

\newtheorem{theorem}{Theorem}

\newtheorem{definition}{Definition}

\newtheorem{lemma}{Lemma}
\newtheorem{assumption}{Assumption}

\newtheorem{remark}{Remark}

\newenvironment{proof}{\noindent {\bf Proof.}}{\hfill \hspace*{1pt} \hfill $\square$\\}

\begin{document}

\title{Robust data-driven control design for linear systems subject to input saturations}
\author{A. Seuret, S. Tarbouriech
\thanks{This work has been partially funded by the European Union - NextGenerationEU. Research by Sophie Tarbouriech is partially funded by ANR via project HANDY, number ANR-18-CE40-0010.}
\thanks{A: Seuret is with the Univ. de Sevilla, Camino de los Descubrimientos, s/n 41092 Sevilla, Spain.  (e-mail: aseuret@us.es).}
\thanks{S. Tarbouriech is with LAAS-CNRS, Universit\'e de Toulouse, CNRS, Toulouse, France. (e-mail: tarbour@laas.fr)..}}

\maketitle

\begin{abstract}
This paper deals with the problem of providing a data-driven solution to the local stabilization of linear systems subject to input saturation. After presenting a model-based solution to this well-studied problem, a systematic method to transform model-driven into data-driven LMI conditions is presented. This technical solution is demonstrated to be equivalent to the recent advanced results on LMI formulations based on S-procedure or Peterson Lemmas. However, the advantage of the proposed method relies on its simplicity and its potential to be applicable to a wide class of problems of stabilization of (non)linear discrete-time systems. The method is then illustrated on an academic example. 
\end{abstract}


\section{Introduction}
The robustness properties for dynamical control systems have been studied in several works from an analysis or design context \cite{ebihara2015s,postlethwaite2007robust, scherer2001theory}. Then, robustness conditions to face uncertainty or presence of additive disturbances have been proposed allowing to ensure the stability and a certain level of performance for the systems under consideration. Most of these conditions are formulated in the form of linear matrix inequalities (LMIs), due to the powerful numerical and optimization procedures, as semi-definite programming, which can be handled. The general common feature of this kind of methods is that they are model-based, possibly modeling also the presence of uncertainties as norm-bounded or polytopic uncertainties. In front of complex systems to model, we need to consider very imprecise or even unknown mathematical models of the dynamic evolution, including for examples nonlinearities. The consequence is then that the  model-based methods for analysis and controller design may reveal to be difficult or even impossible to apply. Adapting the tools issued from the control theory, some work have emerged based on some data information: see, for example, \cite{berberich2020robust,berberich2020data,hou2013model,van2020data} and the references therein.

Hence, some studies addressing the stability or stabilization criteria from a data-driven point of view, have been published. In particular several problems related to the design of state/output feedback controllers for linear systems have been revisited   \cite{dePersis2019formulas,vanWaarde2020noisy,van2020data, bre:per:for:tes/CDC2021}. In these works considering the case of exact data experiments for linear time-invariant systems, i.e. without noise or uncertainties, equivalent formulations between model-based and data-based criteria both in stability analysis and control design have been exhibited. 

Another important feature when dealing with dynamical control systems pertains the control input saturation due to limitations on the actuators \cite{tar:gar:gom:que/book}. At our knowledge, few works in the literature addressed the problem to deal with constrained control (see, for example, \cite{pig:for:bem/CST2018} and the references therein) and more especially with saturating control input. The objective of the current paper is to bring some preliminary bricks in considering this aspect. To this end, we consider the design of a state-feedback saturated control law allowing to ensure the regional (local) asymptotic stability of the origin when the system is noisy-free and to ensure the convergence to an attractor when the system is affected by noise. Indeed, the characterization of an inner-approximation of the basin of attraction of the origin together with an outer-approximation of the attractor are proposed first by model-based techniques and then by data-driven techniques. Taking inspiration from the model-based conditions (see Theorem \ref{th1} in Section \ref{sec:model-theorem}), which are written in a friendly-data form, the data-driven local stabilization is formulated through matrix inequalities conditions (See Theorem \ref{th:theo_data-robust} in Section \ref{sec:data-theorem}). The technique allowing to exhibit the sufficient conditions is based on the use of Lyapunov arguments, generalized sector-bounded conditions to deal with the saturation and S-procedure to handle the presence of noise. A first attempt in this direction was proposed in \cite{seuret2023LCSS}, to revisit the well-know problem of $\mathcal L_2$ stabilization but using a data-driven approach. Differently from this paper, the implicit objectives are here to maximize the inner-approximation of the basin of attraction of the origin and to minimize the outer-approximation of the attractor. The main rational behind the matrix inequalities formulation is due to the friendly form of the model-based conditions which allows to directly derive the data-driven conditions
thanks to a matrix-constrained relaxation (see Lemma \ref{lem0} in Section \ref{Sec:Prel}). Note that the matrix inequalities conditions are quasi-LMIs in the sense that there is the product between matrices and scalars.

The paper is organized as follows. Section \ref{sec:prob} presents the system under consideration and formally state the control problem. In Section \ref{sec:model-theorem}, sufficient conditions  to solve the control problem are formulated as quasi-LMIs, in the sense that there is a product between a matrix and a scalar. Section \ref{sec:main} deals with the local stabilization from a data-driven  point of view. Section \ref{sec:sim} illustrates the theoretical results and proposes some insights regarding the influence of the tuning parameters. Comparison between model-based and data-driven approaches is proposed for different collections of data to depict the trade-off between the estimate of the basin of attraction and of the attractor. Some concluding remarks end the paper in Section~\ref{sec:conclu}.

 \textbf{Notation.} Throughout the paper, $\mathbb N$ denotes the set of natural numbers, $\mathbb R$ the real numbers, $\mathbb R^n$ the $n$-dimensional Euclidean space, $\mathbb R^{n\times m}$ the set of all real $n \times  m$ matrices and $\mathbb S^n$ ($\mathbb S^n_+$) the set of symmetric (positive definite) matrices in $\mathbb R^{n\times n}$.  For any $n$ and $m$ in $\mathbb N$, matrices $I_n$ and $\boldsymbol{0}_{n,m}$ ($\boldsymbol{0}_{n}=\boldsymbol{0}_{n,n}$) denote the identity matrix of $\mathbb R^{n\times n}$ and the null matrix of $\mathbb R^{n\times m}$, respectively. When no confusion is possible, the subscripts of these matrices that precise the dimension, will be omitted. For any matrix $M$ of $\mathbb R^{n\times n}$, the notation  $M\succ0$, ($M\prec0$) means that $M$ is in $\mathbb S^n_+$. 
 For any matrices $A=A^{\!\top},B,C=C^{\!\top}$ of appropriate dimensions, matrix $\left[\begin{smallmatrix}A&B\\\ast & C  \end{smallmatrix} \right]$ denotes the symmetric matrix $\left[\begin{smallmatrix}A&B\\ B^{\!\top}& C  \end{smallmatrix} \right]$. For any matrix $N\in\mathbb R^{n\times m}$, notation $N_{(i)}$, for any $i=1,\dots,n$, stands for the $i^{th}$ row of $N$. $\Vert x \Vert$ denotes the Euclidean norm of $x$. 
 For a symmetric matrix, $\lambda_{m}(M)$ and $\lambda_M(M)$ denote its minimal and maximal eigenvalues respectively of a square matrix $M$. For a matrix $M\in\mathbb S^n$, $M\succ 0$ and any positive scalar $\alpha \in\mathbb R_{>0}$, we denote the ellipsoid $\mathcal E(M,\alpha)=\left\{
x\in\mathbb R^n,~ x^{\!\top} M x\leq \alpha^{-1} \right\}$.

\section{Problem formulation}
\label{sec:prob}

\subsection{System data}
In this paper, we consider the discrete-time linear system subject to an input saturation and affected by an external perturbation. Such a system is described by the following equations 
\begin{align}\label{eq:model_x}
	\left\{
		\begin{array}{lcl} 
			x^+ &=& A x + B \sat (u)+w,\\
			x_0&\in& \mathbb R^{n_x},
		\end{array}
	\right.
\end{align}
where $ x \in \mathbb R^{n_x}$ is the state vector, which adopts the following notation $x^+ = x_{k+1}$ and $x = x_k$, $u=u_k\in \mathbb R^{n_u}$ is the control input. The system is perturbed by the unknown noise signal $w=w_k$, assumed to be bounded leading to the following assumption.
\begin{assumption}\label{ass:noise}
There exists $\lambda >0$ such that the norm of the disturbance verifies $w^{\!\top}_k w_k < \lambda$, for all $k\geq0$. In other words, the disturbance belongs to the following set $\Omega_\lambda$ 
\begin{equation}
    \Omega_\lambda:=\{v\in\mathbb R^n,\quad v^{\!\top} v\leq \lambda \}.
\end{equation}
\end{assumption}

The model $\mathcal G$ is defined from the matrices of the system 
as $\mathcal G:=\begin{pmatrix} A,B\end{pmatrix}\in \mathbb R^{{n_x}\times {n_x}}\times \mathbb R^{{n_x}\times {n_u}}$.

The saturation function $\sat(u)$ is the classical decentralized vector-valued saturation map from $R^{{n_u}}$ to $R^{{n_u}}$, whose the components are defined by:
\begin{equation}\label{sat}
\sat(u_i) = \sgn(u_{i})\min (|u_{i}|,\bar{u}_{i}), \quad \forall i=1,\dots,{n_u}
\end{equation}
where $u_{i}$ refers to the $i^{th}$ control input and $\bar{u}_{i}$ is the $i^{th}$ entry  of the vector $\bar{u} \in \mathbb R^{{n_u}}$ that is the level vector of the saturation.

In this paper, we focus on the design of a static state-feedback control law of the form 
\begin{equation}\label{ControlLaw}
u=Kx,
\end{equation} 
with $K\in\mathbb R^{n_u\times n_x}$, that regionally (locally) stabilizes the trajectories of the closed-loop system \eqref{eq:model_x} with \eqref{ControlLaw} in the absence of perturbation.

\subsection{Preliminaries on generalized sector conditions}

Due to the presence of the saturation map in system (\ref{eq:model_x}), one has to take care of the notion of stability we can ensure for the closed-loop system (\ref{eq:model_x})-\eqref{ControlLaw} \cite{tar:gar:gom:que/book}. In order to address this problem and therefore to handle the presence of the saturation, we follow the 
 approach proposed in \cite{da2005antiwindup} and \cite{tar:gar:gom:que/book}, which consists in rewriting \eqref{eq:model_x} as
\begin{equation}\label{sys}
x^+=(A+BK)x+B\phi(u)+w,
\end{equation}
where $\phi(u)$ is the dead-zone function 
defined by
\begin{equation}\label{deadzone}
\phi(u)=\sat(u)-u.
\end{equation}
Following Remark 7.3 in \cite{tar:gar:gom:que/book}, we use the following \textit{Generalized sector condition} lemma
to handle the dead-zone \eqref{deadzone}.
\begin{lemma} (\cite{tar:gar:gom:que/book})
\label{lema_cond_setor}
Consider 
a matrix $G \in \mathbb R^{n_u \times n_x}$. The following relation holds
\begin{equation}\label{cs}
\phi(u)^{\!\top} 
 T [\sat(u)+G x]\leq 0,
\end{equation}
with any diagonal positive definite  matrix $T 
\in \mathbb 
R^{{n_u}\times {n_u}}$ provided that $x \in \mathcal{S}(G)$ where
\begin{equation}\label{S_set}
\mathcal{S}(G)=\{x \in \mathbb R^{{n_x}}:|G_{(i)} x|\leq \bar{u}_i,\quad \forall i=1,\dots,n_u\}.
\end{equation}
\end{lemma}

\subsection{Control objectives}
It is well-known (see \cite{hu:lin/book2001}, \cite{tar:gar:gom:que/book} and the references therein) that ensuring global asymptotic stability of the origin for saturated systems (\ref{eq:model_x}) (without perturbation) is in general impossible. In other words, one cannot ensure the global asymptotic stability of the origin for any initial condition except if the open loop is not exponentially unstable, and in particular if the open loop is Hurwitz. That means that the basin of attraction of the origin is not the whole state space and needs to be determined, that turns out to be a complex task. In this case, the objective is to characterize an approximation of the basin of attraction of the origin. Considering level sets built from Lyapunov functions leads to inner-approximation of the basin of attraction of the origin. Furthermore, in presence of perturbation satisfying Assumption \ref{ass:noise} (that is bounded but not vanishing noise $w$ in the dynamics), the closed-loop trajectories will not converge to the origin but to some equilibrium points and therefore one can be interested by characterizing the attractor toward which will converge the closed-loop trajectories. Using also level sets built from Lyapunov functions leads to outer-approximation of the real attractor. 

Hence, this paper deals with the following problems: 
\begin{itemize}
      \item[(P1)] In the case $w=0$, characterize an approximation of the basin of attraction, expressed as a level of the Lyapunov function given by $\mathcal E(P,1)$, where $P$ is a symmetric positive definite matrix to be designed. Thus, for any initial condition belonging to $\mathcal E(P,1)$, the closed-loop trajectories will asymptotically converge to the origin.
      
      \item[(P2)] In the case $w \neq 0$, characterize an approximation of the attractor capturing the closed-loop trajectories expressed as a level set of of the same Lyapunov function given by $\mathcal E(P,\varepsilon)$, where $P$ is a symmetric positive definite matrix and $\varepsilon>0$ are to be designed.
      
       \item[(P3)] In addition to the previous problems, maximize the approximation of the basin of attraction $\mathcal E(P, )$, as well as minimize the approximation of the attractor $\mathcal E(P,\varepsilon)$.
 \end{itemize}
 
 Moreover, as we consider that the model is not assumed to be perfectly known but only approached via finite set of data experiments, the final objective is to solve Problems (P1), (P2) and (P3) but from a data-driven point of view. Hence, we want to revisit the solutions obtained from a model-based approach to provide a new data-driven practical stabilization criterion, which arises from matrix manipulations of a suitable model-based preliminary result. This transformation is made possible thanks to a new lemma, which can be seen as a particular case of those from the literature, but which is particularly well-suited for data-driven design.

\section{Model-based local practical stabilization}
\label{sec:model-theorem}
The following theorem follows the methods presented in Chapter 3 in \cite{tar:gar:gom:que/book} and deals with the model-based design of stabilizing controllers (\ref{ControlLaw}) for system \eqref{eq:model_x} under Assumption \ref{ass:noise}.
Therefore, the following theorem states a solution to solve Problems (P1), (P2) and (P3) through a model-based approach.


\begin{theorem}\label{th1}
Under Assumption \ref{ass:noise}, and for a given $\mu\in(0,1)$, $\alpha_1>0$ and $\alpha_2>0$, assume that there exists  
$$\begin{array}{lcl}
\mathcal {D}_{\mathcal V}^1&:=&\{\varepsilon, \mu, W,S,Y,Z\} \\
&\in & \mathbb R_{>0} \!\times\! \mathbb R_{>0}
\!\times\! \mathbb S_{n_x}
\!\times\! \mathbb D_{n_x}\!\times\! {\mathbb R^{{n_x}\!\times\! {n_u}}} \!\times\! {\mathbb R^{{n_x}\!\times\! {n_u}}}
\end{array}$$ solution to the following optimization problem
\begin{equation}\label{opt1}
    \begin{array}{lcl}
\max_{\mathcal {D}_{\mathcal V}^1} & \alpha_1 \varepsilon +\alpha_2 Tr(W)\\
s.t. & \varepsilon >1,\ \Phi(\mathcal G)\succ0,\ \begin{bmatrix}
W & Z_{(i)}^{\!\top}\\
\ast & \bar u_{i}^2 \\ 
\end{bmatrix}\succ0, 
\end{array}    
\end{equation}
for all $i=1,\dots, {n_u}$, where 
$$\begin{array}{lcl}
\Phi(\mathcal G)&=&\begin{bmatrix}
(1-\mu)W & Y^{\!\top}\!\! +\!Z^{\!\top} & W A^{\!\top} \!\! +\!Y^{\!\top} B^{\!\top}\\
\ast &  2S &   SB^{\!\top}\\
\ast &\ast&  W-\frac{\lambda\varepsilon}{\mu} I_{n_x}
\end{bmatrix}.
\end{array}
$$
Then, the control law \eqref{ControlLaw} with $K=YW^{-1}$ ensures that Problems (P1), (P2) and (P3) are solved, i.e. 
\begin{itemize}
    \item When $w=0$, the ellipsoid $\mathcal E(W^{-1},1)$ is an approximation of the basin of attraction of the origin for the closed-loop system \eqref{eq:model_x}-\eqref{ControlLaw}.
\item When $w\in\Omega_\lambda $ and $w\neq 0$, the solutions to the closed-loop system \eqref{eq:model_x}-\eqref{ControlLaw} initialized in $\mathcal E(W^{-1},1)\setminus \mathcal E(W^{-1},\varepsilon)$ converge to the attractor $\mathcal E(W^{-1},\varepsilon)$.
\item Sets $\mathcal E(W^{-1},1)$ and $\mathcal E(W^{-1},\varepsilon)$ are optimal with respect to the cost function $\alpha_1 \varepsilon +\alpha_2 Tr(W)$, which aims deriving a compromise between enlarging $\mathcal E(W^{-1},1)$ and minimizing $\mathcal E(W^{-1},\varepsilon)$.
\end{itemize}
\end{theorem}

\begin{proof}
Consider the Lyapunov function given by $V(x)=x^{\!\top} P x$, for a given $P\in \mathbb S_+^{n_x}$. 
In order to address Problems (P1) and (P2), the objective is to ensure that the forward increment of $V$, i.e. $
    \Delta V(x)={x^+}^{\!\top} Px^+-x^{\!\top} Px$ verifies:
\begin{equation}\label{proof:DV2}
 \Delta V(x) \! \preceq \!0,\ \forall (x,w) \mbox{ s.t.}\left\{\!\!\begin{array}{ll}
     x\!\in\! \mathcal E(P,1), &\! (\mbox{i.e. }x^{\!\top} P x \leq 1) \\
     x\!\notin\! \mathcal E(P,\varepsilon), &\! (\mbox{i.e. } x^{\!\top} P x \geq \varepsilon^{-1})\\
     w\!\in\! \Omega_\lambda, &\! (\mbox{i.e. } w^{\!\top} w \leq \lambda)
\end{array}\right.
\end{equation}
which is implied, thanks to the use of S-procedure \cite{Boyd}, by the following condition
\begin{equation}\label{proof:DV3}
\begin{array}{lcl}
   \mathcal{L}(x,w) &\!\!\!\!\!=& \!\!\!\!\Delta V(x) \!+\! \bar\mu_1 (x^{\!\top} P x \!-\! \varepsilon^{-1})
   \!+\! \bar\mu_2 (\lambda \!-\! w^{\!\top} w) \\
   &&\!\!\!+ \bar\mu_3 (1 -x^{\!\top} P x) \preceq  0
   \end{array}
\end{equation}
with $\bar\mu_i \geq 0$, $i=1, 2, 3$. Let us introduce the following augmented vector $\xi=(x,\phi(Kx),w)$, then $\Delta V(x)$ reads:
\begin{equation}
    \Delta V(x)=\xi^{\!\top} \left(
    \begin{bsmallmatrix}
    (A+BK)^{\!\top}\\B^{\!\top}\\I_{n_x}
    \end{bsmallmatrix}  \!\!\!  P\begin{bsmallmatrix}
    (A+BK)^{\!\top}\\B^{\!\top}\\I_{n_x}
    \end{bsmallmatrix}^{\!\top}
    -\begin{bsmallmatrix}
    P&0&0\\0&0&0\\0&0&0
    \end{bsmallmatrix}
    \right)
    \xi\\
\end{equation}

The satisfaction of the last inequalities in  (\ref{opt1}) means that the ellipsoid $\mathcal E(W,1)$ is included in the set $\mathcal{S}(G)$. From Lemma \ref{lema_cond_setor}, for $x \in \mathcal E(W,1)$ 
inequality (\ref{cs}) holds and reads using $\xi$ as follows:
\begin{equation}\label{cs2}
    \xi^{\!\top} \begin{bsmallmatrix}
    0 & (K+G)^{\!\top} T &0 \\
    T (K+G) & 2T      &0 \\
    0&0&0
    \end{bsmallmatrix}\xi \leq 0.
\end{equation}

Then, merging (\ref{cs2}) and (\ref{proof:DV3}) yields 
$$\mathcal{L}_1(x,w) := \mathcal{L}(x,w) - \xi^{\!\top} \begin{bsmallmatrix}
    0 & (K+G)^{\!\top}T &0 \\
    T (K+G) & 2T      &0 \\
    0&0&0
    \end{bsmallmatrix}\xi\geq \mathcal{L}(x,w).$$
Therefore, having $\mathcal{L}_1(x,w)$ negative  implies that $\mathcal{L}(x,w)$ is also negative for all $(x,w)$. Let us rewrite the expression of $\mathcal{L}_1(x,w)$ in a more compact form.
$$\mathcal{L}_1(x,w) := -\xi^{\!\top} \Phi_1(\mathcal G)\xi -\bar\Theta _\lambda
$$
where
$$
\begin{array}{rcl}
\Phi_1(\mathcal G)&\!\!\!\!=&\!\!\!\!\!\begin{bsmallmatrix}
    \!\!(1\!-\!\bar\mu_1\!+\!\bar\mu_3) P&(K\!+\!G)^{\!\top}T&\!0\\
    T(K\!+\!G)&2T&\!0\\
    0&0&\!\!\!\!\bar \mu_2 I_{n_x}
    \end{bsmallmatrix}\!-\!\begin{bsmallmatrix}
\!    (A\!+\!BK)^{\!\top}\!\!\\B^{\!\top}\\I_{n_x}
    \end{bsmallmatrix}  \!  P\begin{bsmallmatrix}
 \!   (A\!+\!BK)^{\!\top}\!\!\\B^{\!\top}\\I_{n_x}
    \end{bsmallmatrix}^{\!\!\top}
\\
\Theta_\lambda&\!\!\!\!=&\!\!\!\!\bar\mu_1\varepsilon^{-1}-\bar\mu_2\lambda-\bar\mu_3
\end{array}
$$
In the next developments, we will show that the existence of a solution to condition $\Phi(\mathcal G)\succ0$ ensures the existence of $\bar\mu_1>0$, $\bar\mu_2>0$ and  $\bar\mu_3>0$ such that $\mathcal L_1(x,w)\leq0$.

Applying the Schur complement to $\Phi_1(\mathcal G)$, 
pre- and post-multiplying it by $\Diag(W,S,I_{2n_x})$, with $W=P^{-1}$ and $S=T^{-1}$, and selecting $Y=KW$ and $Z=GW$ yield
$$
\begin{bmatrix}
    (1\!-\!\bar\mu_1\!+\!\bar\mu_3) W& Y^{\!\top}\!\!+\!Z^{\!\top} &0 & WA^{\!\top} \!+\!Y^{\!\top} B^{\!\top}\\
    \ast&2S&0 & SB^{\!\top}\\
    \ast&\ast&\bar\mu_2I_{n_x}&I_{n_x}\\
    \ast&\ast&\ast& W
    \end{bmatrix}
    \succ 0,
$$
which is equivalent to 
\begin{equation}
\Phi_2(\mathcal G):=\begin{bmatrix}
    (1\!-\!\bar\mu_1\!+\!\bar\mu_3) W& Y^{\!\top}\!\!\!+\!\!Z^{\!\top} & WA^{\!\top} \!+\!Y^{\!\top}\! B^{\!\top}\\
    \ast&2S& SB^{\!\top}\\
    \ast&\ast& W-\bar\mu_2^{-1} I_{n_x}
    \end{bmatrix}\succ 0
\end{equation}
We have shown so far that having $\mathcal L_1\leq0$ is equivalent to $\Phi_2(\mathcal G)\succ0$ and $\Theta_\lambda >0$. Enforcing the introduction of $\Phi(\mathcal G)$, we get the following expression \begin{equation}\label{ineq:phiGW}
\Phi_2(\mathcal G)\!=\!\Phi(\mathcal G)\!+\!\begin{bmatrix}
    (\mu\!-\!\bar\mu_1\!+\!\bar\mu_3) W& 0& 0\\
    \ast&0& 0\\
    \ast&\ast&\left(\frac{\lambda\varepsilon}{\mu}\!-\!\frac{1}{\bar\mu_2} \right) I_{n_x}
    \end{bmatrix}.
\end{equation}
Then, consider the following selection 
$$\begin{array}{lcl}
\bar\mu_1:=\mu \in (0,1), \quad 
\bar\mu_2:= \frac{\mu}{\lambda\varepsilon}\left(1-\frac{\varepsilon\gamma}{\mu}\right), \quad 
\bar\mu_3<\gamma, 
\end{array} $$ 
 for sufficiently small $\gamma>0$ such that $\bar\mu_2>0$. This selection ensures that $\Theta_\lambda=\gamma-\bar\mu_3>0$, which is required. 
Using this selection for $\bar\mu_2$, we note that 
$$
\frac{\lambda\varepsilon}{\mu}\!-\!\frac{1}{\bar\mu_2}\!=\!\frac{\lambda\varepsilon}{\mu}\left(1\!-\!\frac{1}{1\!-\!\gamma\varepsilon/\mu}\right)\!=\!  \frac{- \gamma \lambda \varepsilon^2}{\mu^2\!-\!\gamma\varepsilon\mu}>  \frac{- \bar\mu_3 \lambda \varepsilon^2}{\mu^2\!-\!\gamma\varepsilon\mu},
$$
so that the following inequality holds
\begin{equation}\label{ineq:phiGW2}
\begin{array}{lcl}
\Phi_2(\mathcal G)\succeq\Phi(\mathcal G)+\bar\mu_3\begin{bmatrix}
     W& 0& 0\\
    \ast&0& 0\\
    \ast&\ast&-   \frac{\lambda \varepsilon^2}{\mu^2\!-\!\gamma\varepsilon\mu} I_{n_x}
    \end{bmatrix}.
    \end{array}
\end{equation}
Hence, if condition $\Phi(\mathcal G)\succ0$ holds, then there exists a sufficiently small $\bar\mu_3$ such that $\Phi_2(\mathcal G)\succ0$.  
%
%
%
%
Therefore, 
condition $\Phi(\mathcal G)\succ0$ ensures that there exist $\bar\mu_1>0,\bar\mu_2>0, \bar\mu_3>0$ such that $\mathcal L_1(x,w)$ is negative, which ensures \eqref{proof:DV2}.

Therefore, the satisfaction of relation  $\Phi(\mathcal G)\succ 0$ 
ensures the local stability of the closed-loop system with the control gain $K=YW^{-1}$, for any initial condition in $x \in \mathcal E (W^{-1},1)$. 

Since $\varepsilon > 1$, one gets $\mathcal E(W^{-1},\varepsilon) \subseteq $ $ \mathcal E(W^{-1},1)$.
In order to prove that $\mathcal E(W^{-1},\varepsilon)$ is an attractor for the closed-loop system, it remains to demonstrate that $\mathcal E(W^{-1},\varepsilon)$ is invariant, i.e.  
$$x\in \mathcal E(W^{-1},\varepsilon)\quad \Rightarrow  \quad x^+\in \mathcal E(W^{-1},\varepsilon).$$ 
To do so, the satisfaction of $\Phi(\mathcal G)\succ0$ ensures that 
$$\begin{array}{lcl}
V(x^+)&\!\!\!\!\!=&\displaystyle \!\!V(x) +\underbrace{\mathcal L_1(x,w)}_{\leq 0}+\underbrace{\xi^{\!\top} \begin{bsmallmatrix}
    0 & (K+G)^{\!\top}T &0 \\
    \ast & 2T      &0 \\
    \ast&\ast&0
    \end{bsmallmatrix}\xi}_{\leq 0}\\
   &\!\!\!\!\!-& \!\!\!\!\! \bar\mu_1 (V(x) \!-\! \varepsilon^{-1}) \!\underbrace{- \bar\mu_2 (\lambda \!-\! w^{\!\top} w)}_{\leq 0} \!- \bar\mu_3 (1 \!-\!V(x))\\
    &\!\!\!\!\!\leq &\!\!\displaystyle (1-\bar\mu_1+\bar\mu_3) V(x) +\!\bar\mu_1 \varepsilon^{-1}-\!\bar \mu_3 \\
\end{array}
$$
Since $\bar\mu_1=\mu$ is in $(0,1)$, there exists a sufficiently small $\bar\mu_3$ such that $1-\bar\mu_1+\bar\mu_3>0$ and, consequently, as $V(x)\leq \varepsilon^{-1}$, we have 
$$
\begin{array}{lcl}
V(x^+)\!\leq\! (1\!-\!\bar\mu_1\!+\!\bar\mu_3)\varepsilon^{-1}\!+\!\bar\mu_1 \varepsilon^{-1}\!-\!\bar\mu_3\!=\!\varepsilon^{-1}\!-\!\bar\mu_3(1\!-\!\varepsilon^{-1}),
\end{array}
$$
which implies that $V(x^+)\!\leq\!\varepsilon^{-1}$, since $ \varepsilon> 1$.

The last step of the proof addresses the particular case when $w=0$ or equivalently $\lambda=0$. In this situation, $\varepsilon$ disappears from the condition and can be selected as large as possible, so that $\mathcal E(W^{-1},\varepsilon)$ shrinks to $\{0\}$.  Therefore, the solution to the system, initialized in $\mathcal E(W^{-1},1)$ converges to the origin.
\end{proof}

\begin{remark}
Note that $\Phi(\mathcal G)\succ 0$ is not an LMI due to the terms $\mu W$ and $\lambda \varepsilon \mu^{-1}$. Nevertheless, once the parameter $\mu$ is fixed a priori, the problem becomes an LMI and can easily be solved. A gridding on the parameter $\mu$ should be included to the LMI solver.
\end{remark}


\section{Local data-driven control design}
\label{sec:main}
The objective of this section is to transform the previous model-based theorem into a data-based one. The basic idea is to use the inherent robustness property of the previous theorem with respect to $(A,B)$, to derive such a transformation. After presenting the main features and assumptions on the data experiments, a new lemma is presented to achieve this transformation in a systematic manner. An application of this lemma to the problem in hand  finally allows us to exhibit the main result of this paper.

\subsection{Data collections and assumption}

Unlike the usual situation where the system model is available, our objective is here to provide a data-driven-based design result. 
To do this, let us first specify the notion of data, we are considering here. We define the following data matrices that collect the available measurement 
for the control design. 
\begin{equation}\label{def_data}
\begin{array}{lclcllcr}
\mathcal X^+&\!\!:=&\![&\! x^+_1& x^+_2&\dots& x^+_p &\!\!],\\
\mathcal X&\!\!:=&\![&\! x_1&x_2&\dots& x_p&\!\!],\\
\mathcal U&\!\!:=&\![&\!\sat(u_1)&\sat(u_2)&\dots&\sat(u_p)&\!\!],\\
\omega&\!\!:=&\![&\! w_1&w_2&\dots& w_p&\!\!].\\
\end{array}
\end{equation}
We define the data collection $\mathcal D$ as follows
\begin{equation}\label{def:data}
    \mathcal D:=(\mathcal X^+,\mathcal X,\mathcal U) \in \mathbb{R}^{n_x \times p} \times \mathbb{R}^{n_x \times p} \times \mathbb{R}^{n_u \times p}.
\end{equation}
According the systems dynamics, these data verify
\begin{equation}\label{eq:info_data}
\mathcal X^+=A \mathcal X+B\mathcal U  +\omega\quad \in\mathbb R^{n_x\times p}.
\end{equation}
for some matrices $(A,B)$. The noise samples $w(0),w(1),$ $\dots,w(p)$ collected in the matrix $\omega$ are unknown. However, the following assumption is made on the ``energy'' of the noise, following the presentation of \cite{bisoffi2022data}.

\begin{assumption}\label{Assumption2}
Assume that there exists a known matrix $\Delta_\omega$ in $\mathbb S^{n_x}_+$ such that the noise samples matrix $\omega$ verify
\begin{equation}\label{def:Omega}
\omega \in\Omega :=\{v \in\mathbb R^{n_x\times p},\ v  v ^{\!\top}\leq p\lambda \Delta_\omega  \}.
\end{equation}
\end{assumption}

Observe in this assumption that matrix $\Delta_\omega$ is independent of the noise amplitude $\lambda$ and of the number of experiments $p$. This assumption can be related to the co-variance of the noise $\omega$ as mentioned in \cite{bisoffi2022data} or \cite{vanWaarde2020noisy}.

The authors of \cite{bisoffi2022data} have interestingly presented an alternative formulation of $\Omega$ using the data collection $\mathcal D$. This method that was also developed in many other papers such as, for example, \cite{vanWaarde2020noisy}, \cite{van2020data},
consists in redefining the following set 
\begin{equation}\label{def:uncertainty}
\mathcal C\! :=\!\{[A  \ B],\ \mathcal X^+\!\!=\!A \mathcal X\!+\!B\mathcal U  \!+\!\omega,\ \forall \omega\in\Omega  \}\subset \mathbb R^{n_x\times (n_x+n_u)} .
\end{equation}
Thus, an alternative representation of $\mathcal C$ can be derived as noticed in \cite{bisoffi2022data}, as follows
\begin{equation}\label{def:uncertaintybis}
\mathcal C \!\!:=\!\left\{\![A  \ B], \begin{bsmallmatrix}
I\\
A^{\!\top}\\
B^{\!\top}
\end{bsmallmatrix}^{\!\!\top}\!\! 
\begin{bsmallmatrix}
\mathcal X^+{\mathcal X^+}^{\!\top}\!\!-p\lambda \Delta_\omega  & -\mathcal X^+\mathcal X^{\!\top}  & -\mathcal X^+\mathcal U^{\!\top} \\
\ast & \mathcal X\mathcal X^{\!\top}  & \mathcal X\mathcal U^{\!\top}\\
\ast & \ast &  \mathcal U\mathcal U^{\!\top}\\
\end{bsmallmatrix}\!
\begin{bsmallmatrix}
I\\
A^{\!\top}\\
B^{\!\top}
\end{bsmallmatrix} \!\!\preceq\!0\!\right\}
\end{equation}


We classically consider that the data are informative as defined below.
 \begin{definition}  The data collection $\mathcal D$ in \eqref{def_data} is said informative  if  matrix $\begin{bsmallmatrix}
\mathcal X\\
\mathcal U\\
\end{bsmallmatrix}\begin{bsmallmatrix}
\mathcal X\\
\mathcal U\\
\end{bsmallmatrix}^{\!\top}$ is non singular.
 \end{definition}


\subsection{Matrix-constrained Relaxation} 
\label{Sec:Prel}
The following lemma, which is 
the main brick of the paper, presents a generic method to transform a problem of a particular matrix inequality which depends on parameters verifying a quadratic constraint into a formulation that is independent of these parameters. It is stated below.

\begin{lemma}\label{lem0}
For given positive integers $n_1,n_2,n_3$, consider matrices $(\mathcal M_1,\mathcal M_2,\mathcal M_3)$ in $\mathbb S^{n_1}_+  \times  \mathbb R^{n_1\times n_2} \times  \mathbb S^{n_3}_+$ and $(\mathcal N_1,\mathcal N_2,\mathcal N_3)$ in $\  S^{n_3}  \times  \mathbb R^{n_3\times n_2} \times  \mathbb S^{n_3}_+$, i.e. with $\mathcal N_3\succ 0$. 

Then, the following statements are equivalent
\begin{itemize}
\item[\textit{(i)}] Inequality 
\begin{equation}\label{lem0_1}
\mathcal M(\mathcal A)=\begin{bmatrix}
\mathcal M_{1} & 
\mathcal M_{2}\mathcal  A\\
\ast & \mathcal M_{3}
\end{bmatrix}\succ 0,\quad \forall \mathcal A\in  \Sigma_\mathcal N,
\end{equation}
holds true where $ \Sigma_\mathcal N$ represents the set of allowable uncertain matrices $\mathcal A$ characterized by a quadratic constraint defined as follows
 \begin{equation}\label{def_Sigma}
\!\!\!\! \Sigma_\mathcal N\!:=\!\left\{
 \mathcal A \in\mathbb R^{n_2\times n_3} , 
 \left[\begin{matrix}
I_{n_3}\\
\mathcal A\\
\end{matrix}\right]^{\!\top}\!\!  \underbrace{\begin{bmatrix}
 \mathcal N_1&\!\!\mathcal N_2\\
 \ast &\!\!\mathcal N_3\\
\end{bmatrix}}_{\mathcal N}\!\left[\begin{matrix}
I_{n_3}\\
\mathcal A\\
\end{matrix}\right]\!\preceq \!0.
 \right\}
 \end{equation}

\item[\textit{(ii)}] There exists $\eta>0$ such that 
\begin{equation}\label{lem:IneqY}
\begin{bmatrix}
\mathcal M_1 & \boldsymbol{0}_{n_1,{n_3}} & \mathcal M_{2} \\
\ast & \mathcal M_3+\eta\mathcal N_1& \eta\mathcal N_2\\ 
\ast & \ast &\eta\mathcal N_3 
\end{bmatrix}\succ 0.
\end{equation}
\end{itemize}
\end{lemma}

The proof of this lemma is postponed to the appendix.

Lemma \ref{lem0} provides an alternative formulation in robust analysis for uncertain matrices subject to quadratic constraints of the form \eqref{def_Sigma} compared to the one presented in \cite{vanWaarde2020noisy,van2020data}. For the sake of consistency, the S-Lemma provided in \cite{vanWaarde2020noisy} is recalled in the following lemma.
\begin{lemma}{\cite[Th.9]{vanWaarde2020noisy}}\label{lem:Waarde}
Let $\mathcal M_s,  \mathcal N_s \in \mathbb R^{(n_s+m_s)\times(n_s+m_s)}$ be symmetric matrices and assume that inequality
$
\begin{bsmallmatrix}
I_{n_s}\\
\mathcal A_s^{\!\top}
\end{bsmallmatrix}^{\!\top}  \mathcal N_s \begin{bsmallmatrix}
I_{n_s}\\
\mathcal A_s^{\!\top}
\end{bsmallmatrix}\succeq0
$ for at least one matrix $\mathcal A_s \in \mathbb R^{n_s\times m_s}$. Then the next statements are equivalent:
\begin{enumerate}
\item[(i)] $\begin{bmatrix}
I_{n_s}\\
\mathcal A_s^{\!\top}
\end{bmatrix}^{\!\top} \mathcal M_s \begin{bmatrix}
I_{n_s}\\
\mathcal A_s^{\!\top}
\end{bmatrix}\succ0,\ \forall \mathcal A_s\in \Sigma_{\mathcal N_s} 
$, \newline where set $\Sigma_{\mathcal N_s}$ has the same definition as in \eqref{def_Sigma} but replacing $\mathcal A$  and $\mathcal N$ by $\mathcal A_s$ and $\mathcal N_s$, respectively.
\item [(ii)] There exists $\eta > 0$ such that $\mathcal M_s - \eta  \mathcal N_s \succ 0$.
\end{enumerate}
\end{lemma}

Note that both lemmas address the problem of the satisfaction of an inequality subject to uncertain matrices characterized by a quadratic constraint. The main interest of both lemmas is to derive equivalent inequalities that are independent of the uncertain matrix $\mathcal A$ (or $\mathcal A_s$). Finally, both lemmas can be seen as application of the usual manipulations on LMI such as Schur Complement, Finsler's lemma and S-procedure.  Apart from presenting these similarities, both lemmas have substantial differences. First, Lemma \ref{lem:Waarde} requires that matrix $\mathcal M_s$ has the same size as $\mathcal N_s$, the matrix that characterizes the quadratic constraint on $\mathcal A$. Lemma~\ref{lem0} is more flexible in this sense, as there is no relationship between matrices $\mathcal M_1$ and $\mathcal N$, which are independent. This flexibility has the benefit of reducing the initial manipulations to derive, from usual stability or control problems, the appropriate expression of $\mathcal M_i$ and $\mathcal N_i$, for $i=1,2,3$, to fit the framework of Lemma~\ref{lem:Waarde}. In fact, the relationship between both lemmas can be seen by selecting $n_s=n_1+n_3$, $m_s= n_2$, $\mathcal A_s^{\!\top} =\begin{bmatrix} 0_{n_2\times n_3}& \mathcal A^{\!\top} \end{bmatrix}$ and 
$$
\begin{array}{c}
\mathcal M_s=\begin{bmatrix}
\mathcal M_1 &
0
&\mathcal M_2\\
\ast &  \mathcal M_3 &0\\
\ast&\ast&0\\
\end{bmatrix}, \quad \mathcal N_s= -\begin{bmatrix}
0&0 &0\\
\ast &\ \ \mathcal N_1 &  \mathcal N_2\\
\ast&\ast& \mathcal N_3\\
\end{bmatrix}.
\end{array}
$$
From this selection, it is clear that item (ii) of Lemma \ref{lem:Waarde} is equivalent to item (ii) of Lemma~\ref{lem0}, showing that Lemma~\ref{lem0} is a particular case of Lemma~\ref{lem:Waarde}. That being said, the main advantage of Lemma~\ref{lem0} is that the structure of $\mathcal M(\mathcal A)$ arises in many LMI problems in control as it will be shown in the next section. Therefore it avoids one from enforcing an initial LMI problem to fit with the structure of Lemma \ref{lem:Waarde}, which is not an easy task in general.

\begin{remark}
A similar lemma was already presented in \cite{Seuret2022Data:SAS}, where block $\mathcal M_2\mathcal A$ is extended to $\bar{\mathcal M}_{2}+\mathcal M_2\mathcal A$, where $\bar{\mathcal M}_{2}$ is a term that is independent of the uncertain matrix. In addition, a deeper discussion on the similarities with the existing lemmas from the literature has been proposed therein. 
\end{remark}
\subsection{Data-driven local stabilization of saturated systems}
\label{sec:data-theorem}

This section provides a new contribution on the data-based design of stabilizing control law for linear systems subject to input saturation. The method is highly inspired from 
\cite{vanWaarde2020noisy} but has been reformulated to get a simpler and user-friendly formulation. In this paragraph, we will demonstrate how Lemma \ref{lem0} can be easily applied to the stabilization problem of saturated systems. The robust local stabilization of system \eqref{eq:model_x} can be formalized in the following theorem. 

\begin{theorem}\label{th:theo_data-robust}
Under Assumptions \ref{ass:noise} and \ref{Assumption2}, i.e. for a given matrix $\Delta_w>0$, and for given $\mu \in (0,1)$, $\alpha_1>0$ and $\alpha_2>0$, assume that there exist 
$$\begin{array}{lcl}
\mathcal {D}_{\mathcal V}^2&:=&\{ \varepsilon,\eta, W,S,Y,Z\} \\
&\in & \mathbb R_{>0} \!\times\! \mathbb R_{>0}
\!\times\! \mathbb S_{n_x}
\!\times\! \mathbb D_{n_x}\!\times\! {\mathbb R^{{n_x}\!\times\! {n_u}}} \!\times\! {\mathbb R^{{n_x}\!\times\! {n_u}}}
\end{array}$$
that are solution to the following optimization problem
\begin{equation}\label{opt2}
    \begin{array}{lcl}
\max_{\mathcal {D}_{\mathcal V}^2} & \alpha_1 \varepsilon +\alpha_2 Tr(W)\\
s.t. & \varepsilon>1, \ \eta>0,\  \Psi(\mathcal D)\succ0,\ \begin{bmatrix}
W & Z_{(i)}^{\!\top}\\
\ast & \bar u_{i}^2 \\ 
\end{bmatrix}\succ0, 
\end{array}    
\end{equation}
for all $i=1,\dots, {n_u}$, where
\begin{equation}
    \begin{array}{rcl}
    \Psi(\mathcal D)&\!\!\!\!=&\!\!\!\!	\begin{bmatrix}	
(1\!-\!\mu)	W & \!\!\!Y^{\!\top}\!\!+\!Z^{\!\top} &\!\!\!\!\!0  &\!\!\!\!W&\!\!\!\! Y^{\!\top}\\
\ast& \!\!\!2S & \!\!\!\!\!0& \!\!\!\!0& \!\!\!\!\! S\\
\ast&\!\!\!\ast& \!\!\!\!\! {\Psi}_3& \!\!-\eta\mathcal X^+\mathcal X^{\!\top} &\!\! -\eta\mathcal X^+\mathcal U^{\!\top}\\ 
\ast & \!\!\! \ast&\!\!\!\!\!\ast& \!\! \eta\mathcal X\mathcal X^{\!\top} & \!\!\eta\mathcal X\mathcal U^{\!\top}\\
\ast & \!\!\!\ast&\!\!\!\!\ast&\!\! \ast&\!\! \eta\mathcal U\mathcal U^{\!\top}
     \end{bmatrix}\\
   {\Psi}_3&\!\!\!\!\!=&\! W-\frac{\lambda\varepsilon}{\mu} I_{n_x}+\eta\mathcal X^+{\mathcal X^+}^{\!\top}-\eta p \lambda \Delta_\omega\\
    \end{array}
\end{equation}
Then, the control law \eqref{ControlLaw} with $K=YW^{-1}$ ensures that Problems (P1) and (P2) are solved, or equivalently, the following statements hold:
\begin{itemize}
    \item When $w=0$, the ellipsoid $\mathcal E(W^{-1},1)$ is an approximation of the basin of attraction of the origin for the closed-loop system \eqref{eq:model_x}-\eqref{ControlLaw};
\item When $w\in\Omega_\lambda $ and $w\neq 0$, the solutions to the closed-loop system \eqref{eq:model_x}-\eqref{ControlLaw} initialized in $\mathcal E(W^{-1},1) \setminus \mathcal E(W^{-1},\varepsilon)$ converge to the attractor $\mathcal E(W^{-1},\varepsilon)$.
\item The estimation of the basin of attraction and of the attractor are optimal with respect to the cost function $\alpha_1 \varepsilon +\alpha_2 Tr(W)$, leading to a compromise between enlarging $\mathcal E(W^{-1},1)$ and minimizing $\mathcal E(W^{-1},\varepsilon)$.
\end{itemize}
\end{theorem}

\begin{proof}
Let us first note that, in Theorem \ref{th1}, 
conditions $\varepsilon > 1$ and $\begin{bsmallmatrix}
W & Z_{(i)}^{\!\top}\\
\ast & \bar u_{i} \\ 
\end{bsmallmatrix}\succ0$ do not involve the matrices of the system and thus remains unchanged. On the other side, condition $\Phi(\mathcal G)\succ0$ has to be modified in order to remove the system's matrices. To do so, let us note that $\Phi(\mathcal G)$ is rewritten as follows
\begin{equation}
\Phi(\mathcal G)
=\begin{bmatrix}
\begin{bmatrix} (1-\mu)W & X^{\!\top}+ Y^{\!\top} \\
\ast& 2S\end{bmatrix}
& \begin{bmatrix}
 W & Y^{\!\top}\\ 
  0& S 
 \end{bmatrix}  \left[\begin{matrix}
A^{\!\top}\\
B^{\!\top}\\
 \end{matrix}\right]\\
\ast &  W -\frac{\lambda\varepsilon}{\mu} I_{n_x}
\end{bmatrix},
\end{equation}
which has the same structure as the one of matrix $\mathcal M(\mathcal A)$ introduced in Lemma~\ref{lem0}. Indeed, selecting $n_1=n_2={n_x}+{n_u}$, $n_3={n_x}$ and 
$$\begin{array}{l}
\mathcal M_1\!=\!\begin{bmatrix} \!(1\!-\!\mu)W & \!\!\!\!X^{\!\top}+Y^{\!\top} \\
\ast& \!\!\!\!2S\end{bmatrix}\!\!,\ \mathcal M_2\!=\!\begin{bmatrix} W & \!\!\!\!Y^{\!\top} \\
0& \!\!\!\!S\end{bmatrix}\!,\\ \mathcal M_3\!=\! W\!\!-\!\frac{\lambda\varepsilon}{\mu} I_{n_x},\\
 \mathcal N_1=\mathcal X^+{\mathcal X^+}^{\!\top}\!\!-\!p \lambda \Delta_\omega, \ \mathcal N_2\!=\!-\mathcal X^+\begin{bmatrix}
\mathcal X\\
\mathcal U\\
\end{bmatrix}^{\!\top}\!\!\!, \
\mathcal N_3\!=\!\begin{bmatrix}
\mathcal X\\
\mathcal U\\
\end{bmatrix}\begin{bmatrix}
\mathcal X\\
\mathcal U\\
\end{bmatrix}^{\!\top}\!\!\!.
\end{array}
$$
with the uncertainty matrix
$$
\mathcal A=\left[\begin{matrix}
A^{\!\top}\\
B^{\!\top}\\
\end{matrix}\right] \in\mathcal C\subset \mathbb R^{(n_x+n_u)\times n_x}.
$$
Altogether, the problem can be expressed as in \eqref{lem0_1} with this set of matrices. 

 Then, the problem resumes to the satisfaction of condition $\Phi(\mathcal G)\succ0$ for all matrices $[A\ B]^{\!\top}=\mathcal A \subset \mathcal C$. Note that $\mathcal N_3\succ 0$ is required in Lemma~\ref{lem0}, which refers to the informativity of the data and is also required for the satisfaction of condition $\Psi(\mathcal D)\succ0$. We are thus in position to apply Lemma~\ref{lem0}, which states that having $\Phi(\mathcal G)\succ 0$ for all $\mathcal A \in\mathcal C$ is equivalent to the existence of $\eta>0$ such that $\Psi(\mathcal D)\succ0$.
\end{proof}

\begin{remark}
Note that the informativity of the data is a necessary condition for the solvability of the LMI optimization problem. In other words, if matrix $\begin{bsmallmatrix}
\mathcal X\\
\mathcal U\\
\end{bsmallmatrix}\begin{bsmallmatrix}
\mathcal X\\
\mathcal U\\
\end{bsmallmatrix}^{\!\top}$ is singular, then the LMI problem cannot be solved.
\end{remark}
\begin{remark}
Note that, contrary to \cite{berberich2020robust}  or \cite{vanWaarde2020noisy}, our method does not require to study the 
dual system involving $(A+BK)^{\!\top}$. First, this trick is not possible for nonlinear systems. Second, our method directly treats the initial problem without using this artifice of calculus to avoid a technical problem.
\end{remark}

\begin{remark}
Differently from \cite{bisoffi2022data}, the additional variable $\eta$ cannot be ``absorbed" by the other decision variables. This due to the problem of local stabilization, for which the scalability of the Lyapunov matrix $W$ is not permitted since it appears multiplied by the scalar $\mu$.
\end{remark}

\begin{remark}
It is worth noting that the resulting LMI condition is very similar to the ones presented in  \cite{bisoffi2022data} based on the application of the Pertersen's lemma for the case of linear systems. Indeed, selecting $S=0$, $Y=Z=0$ and $\varepsilon=0$, the same LMI condition as the one of Th.2 in \cite{bisoffi2022data} is retrieved. This makes us think that Lemma \ref{lem0} can be related to the Petersen's lemma as well as we have showed the link with the S-Procedure of \cite{vanWaarde2020noisy}. Again, the advantage of Lemma \ref{lem0} over both these lemmas is that the structure of many LMI problems in control fits with the structure of Lemma \ref{lem0}, so that no additional manipulation is required.
\end{remark}

\section{Numerical application}\label{sec:sim}

Consider the discrete-time systems \eqref{sys} borrowed from \cite{tarbouriech1997admissible,vassilaki1988feedback}, with the following matrices
\begin{equation}
    A=\begin{bmatrix}
    0.8 & 0.5\\ -0.4 &1.2
    \end{bmatrix},\quad B=\begin{bmatrix}
    0\\ 1
    \end{bmatrix},\quad \bar u_1 = 5.
\end{equation}

\subsection{Model-based solution}
This section aims at illustrating the model-based stabilization result of Theorem \ref{th1}. Solving the optimization problem \eqref{opt1} with $\alpha_1=1$ and $\alpha_2=10^{-3}$ for $\lambda =0.01, 0.05$ and $0.1$ and $\mu=0.08,0.3$ and $0.6$. Figure \ref{fig:ex1:model} shows the estimation of the basin of attraction (dashed blue line) and of the attractor (black line) for all combinations of $(\lambda,\mu)$. In addition, each sub-figure also depicts $40$ simulations that are initiated at the boundary of $\mathcal E(W^{-1},1)$, which all converge the attractor $\mathcal E(W^{-1},\varepsilon)$. For the particular case $(\lambda,\mu)=(0.05,0.4)$, we have obtained $W=\begin{bsmallmatrix}
78.67 & -14.16\\ -14.16 & 27.09
\end{bsmallmatrix}$ and $\varepsilon=79.54$.

\begin{figure}[hhh!]
\hspace{-0.42cm}
\includegraphics[width=0.54\textwidth]{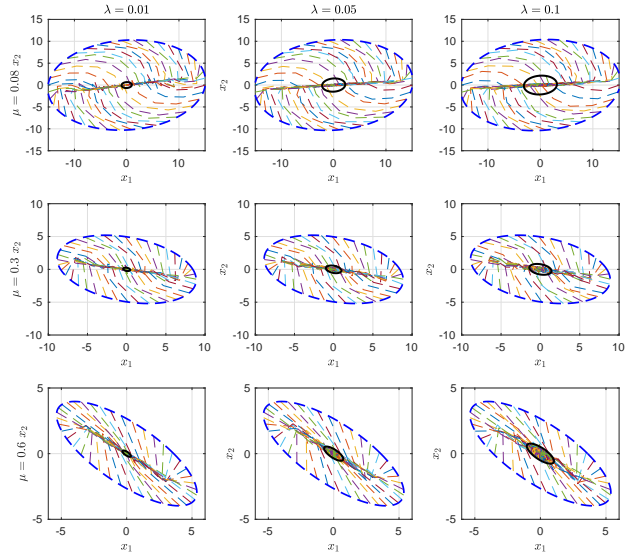}
\vspace{-1cm}
	\caption{Estimations of the basin of attraction and of the attractor obtained by solving the model-based optimization problem \eqref{opt1} for various $\lambda$ and $\mu$.}
	\label{fig:ex1:model}
\end{figure}

One can observe on Figure \ref{fig:ex1:model} that increasing the magnitude 
of the noise (i.e., $\lambda$), increases the size of the attractor, while the size of 
approximation of the basin of attraction remains almost the same. 
Interestingly, Figure \ref{fig:ex1:model} shows that the tuning parameter $\mu$ has an important effect on the solution. Indeed, for small values of $\mu$, the approximation of the basin of attraction is larger, but the approximation of the attractor is rather poor, since the black ellipsoids are very large compared to the chattering around the origin. Reversely, when $\mu$ is large, the optimization process provides a more accurate approximation of the attractor but at the price of drastically reducing the size of the approximation of the basin of attraction. This shows that the tuning parameter $\mu$ plays the role of the balance between the optimization of estimations of the basin of attraction and of the attractor, which cannot be performed simultaneously. This compromise can be expected from the LMI condition, since having $(1-\mu)$ small, implies that $W$ has to be large (block $(1,1)$ of $\Phi(\mathcal G)$) and having $\mu$ small implies that $W\varepsilon^{-1}$ has to be large (block $(3,3)$ of $\Phi(\mathcal G)$).  

\begin{figure*}
\centering
     \begin{subfigure}[b]{0.45\textwidth}
        \hspace{-1.5cm} 
         \includegraphics[width=1.22\textwidth]{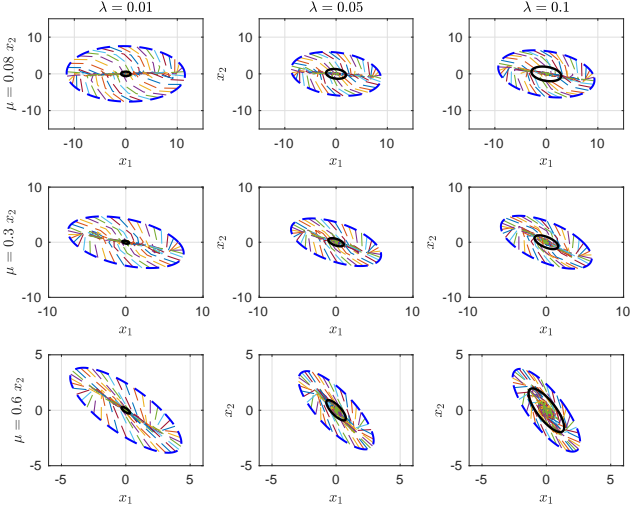}
         \vspace{-0.8cm}
         \caption{Simulations obtained with a collection of $p=5$ data.}
         \label{fig:data_p5}
     \end{subfigure}
     \begin{subfigure}[b]{0.45\textwidth}
       \hspace{-0.6cm}  
         \includegraphics[width=1.22\textwidth]{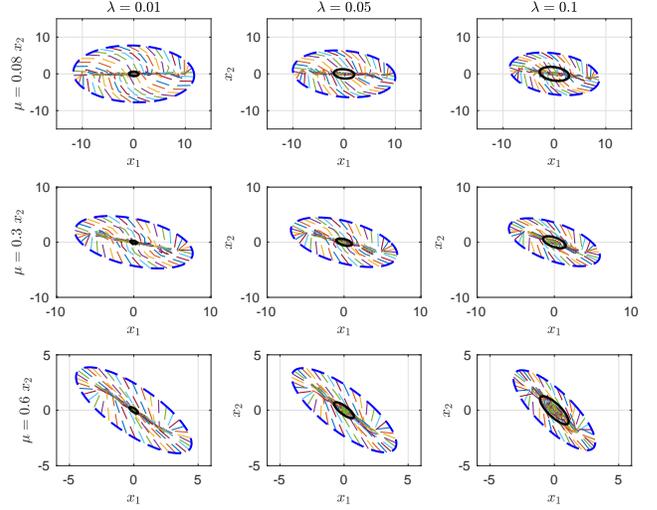}
         \vspace{-0.8cm}
         \caption{Simulations obtained with a collection of $p=20$ data.}
         \label{fig:data_p20}
     \end{subfigure}
        \caption{Estimations of the basin of attraction and of the attractor obtained by solving the data-based optimization problem \eqref{opt1} for various $\lambda$ and $\mu$.}
        \label{fig:data}
\end{figure*}

\subsection{Data-driven local stabilization}

This section illustrates the impact of the data-driven approach compared to the model-driven one. Let us first present the construction of the data. Following the dynamics of the system, we have selected $p=5$ and $p=20$ random values of $x_i$ and $u_i$ in \eqref{def_data} where each component of $x_i$ lies in $[-1,1]$. Then, we have created the random noise signal $w_i$ so that the collection of noise vector $\omega$ belongs to $\Omega$ as described in Assumption~\ref{Assumption2} with $\Delta_\omega=0.05I$. Altogether, vector $\mathcal X^+$ has been computed using equation \eqref{eq:info_data}.

The results are illustrated in Figure~\ref{fig:data}. More precisely, Figures \ref{fig:data_p5} and \ref{fig:data_p20} illustrate the results obtained by solving the data-driven stabilization conditions of Theorem \ref{th:theo_data-robust} with $p=5$ and $p=20$ experiments, respectively. As for the model-based condition of Theorem~\ref{th1}, one can see that increasing the magnitude of the noise leads to a reduction of the size of the estimation of the basin of attraction (dashed blue ellipsoids), and also an increase of the size of the attractor (black ellipsoids). The same effect of the tuning parameters can be also detected. 

Compared to the model-driven results in Figure~\ref{fig:ex1:model}, it occurs a notable reduction of both ellipsoids for each case. This is due to the fact that the model-based solution has no uncertainties in the matrices $A$ and $B$, while the data-driven condition includes the inherent uncertainties due to the noise affecting the data. 

Comparing now both data-driven result, one can see that each estimation of the basin of attraction obtained for $p=5$ are smaller than the one derived with $p=20$. This can be interpreted by the fact that increasing the number of data experiments give more information on the system so that the uncertainties due to the noise are reduced. In other words the set of allowable matrices in $\mathcal C$ becomes ``smaller'' as the number of data increases.

As a last comment, it seems surprising that the second column  of Figure\ref{fig:data_p5} (corresponding to $\lambda=0.05$) shows smallest estimation of the basin of attraction than the third column (corresponding to $\lambda=0.1$). In light of the previous paragraph, a possible explanation is again related to the set of allowable matrices $\mathcal C$. As the data have been generated randomly, a possible explanation is that set $\mathcal C$ obtained with a random noise signal in $\Omega$ with $\lambda =0.05$ is ``smaller'' than the one obtained with $\lambda =0.1$. 

\section{Conclusion}
\label{sec:conclu}
This paper addressed the problem of providing a data-driven solution to the local stabilization of linear systems subject to input saturation. Using model-based solution to this well-studied problem, a systematic method to transform model-driven into data-driven LMI conditions is presented thanks to some adequate rewriting of the conditions. Although this technical solution is shown to be equivalent to some recent advanced results, its main advantage relies on its simplicity and its potential to be applicable to a broad class of problems of stabilization of (non)linear discrete-time systems. 

The proposed results pave the way for future work, as in particular the possibility to consider more complex dynamics or isolated nonlinearities

\bibliographystyle{plain}
\bibliography{bibSat_data.bib}          


\section*{Appendix-Proof of 
Lemma \ref{lem0}}
\begin{proof}
The proof is divided into two steps. 
\newline \underline{\textit{(i)$\Rightarrow$(ii):}} The first step of the proof is to find an appropriate expression of matrices $\mathcal A$ that belong to $\Sigma_\mathcal N$. 
For any matrix $\mathcal A$ in $\Sigma_\mathcal N$, it holds
$$\begin{array}{lcl}
0&\!\!\preceq&\!\! 
\mathcal R-\begin{bmatrix}
I_{n_3}\\
\mathcal A\\
\end{bmatrix}^{\!\top} 
\begin{bmatrix}
\mathcal R+ \mathcal N_1&\!\!\mathcal N_2\\
 \ast &\!\!\mathcal N_3\\
\end{bmatrix}   \left[\begin{matrix}
I_{n_3}\\
\mathcal A\\
\end{matrix}\right]\\
\end{array}
$$
where $\mathcal R$ is any matrix in $\mathbb S^{n_3}$. In addition,  $\mathcal N_3\succ0$ ensures that there exists a matrix $\mathcal R$ such that $ \begin{bsmallmatrix} \mathcal R+ \mathcal N_1&\!\!\mathcal N_2\\
 \ast &\!\!\mathcal N_3 \end{bsmallmatrix}$ is positive definite. That  allows applying the Schur complement as follows to obtain
$$
\begin{bmatrix}
\begin{bmatrix}
\mathcal R+ \mathcal N_1&\!\!\mathcal N_2\\
 \ast &\!\!\mathcal N_3\\
\end{bmatrix}^{-1} & \
\left[\begin{matrix}
I_{n_3}\\
\mathcal A\\ 
\end{matrix}\right]\\
\ast& \mathcal R
\end{bmatrix}\succ0
$$
Note that it is not the usual way to apply the Schur Complement, but this dual way has been considered to keep block $(1,2)$ (resp. $(2,1)$) with $ \left[\begin{smallmatrix}
I_{n_3}\\
\mathcal A\\ 
\end{smallmatrix}\right]$ (resp. its transpose). Next, pre- and post-multiply the previous inequality by $$\begin{bmatrix}
 \boldsymbol{0}_{n_1,{n_3}} & \mathcal Z &\boldsymbol{0}_{{n_1},{n_3}}\\ \boldsymbol{0}_{n_3} & \boldsymbol{0}_{{n_3},{n_2}}&I_{n_3} \end{bmatrix}$$ and its transpose, respectively, where $\mathcal Z$ is any matrix in $\mathbb R^{n_1\times {n_2}}$, i.e. of the same dimensions as $ \mathcal M_2$. This yields that having $\mathcal A\in  \Sigma_\mathcal N$ implies 
\begin{equation}\label{Lem0_Proof_cond}
\begin{array}{l}
\begin{bmatrix}
\begin{bmatrix} \boldsymbol{0}_{{n_1},{n_3}} \\ \mathcal Z^{\!\top}  \end{bmatrix} ^{\!\!\top}\!\!
\begin{bmatrix}
\mathcal R+ \mathcal N_1&\!\!\mathcal N_2\\
 \ast &\!\!\mathcal N_3\\
\end{bmatrix}^{-1}\!\! \begin{bmatrix} \boldsymbol{0}_{n_1,{n_3}} \\ \mathcal Z^{\!\!\top}  \end{bmatrix}^{\!\!\top}
&\!\!\!\! \begin{bmatrix} \boldsymbol{0}_{{n_1},{n_3}} \\ \mathcal Z^{\!\top}  \end{bmatrix} ^{\!\!\top}\!\!
\left[\begin{matrix}
I_{n_3}\\
\mathcal A\\
\end{matrix}\right]\\
\ast&\!\!\!\! \mathcal R
\end{bmatrix}\\
=\begin{bmatrix}
\begin{bmatrix} \boldsymbol{0}_{n_1,{n_3}} \\ \mathcal Z^{\!\top}  \end{bmatrix} ^{\!\!\top}
\begin{bmatrix}
\mathcal R+ \mathcal N_1&\!\!\mathcal N_2\\
 \ast &\!\!\mathcal N_3\\
\end{bmatrix}^{-1} \begin{bmatrix} \boldsymbol{0}_{{n_1},{n_3}} \\ \mathcal Z^{\!\top}  \end{bmatrix} 
^{\!\top} &  \mathcal Z
\mathcal A\\
\ast& \mathcal R
\end{bmatrix}\succ0.
\end{array}
\end{equation}
The previous calculations ensure that inequality \eqref{lem0_1} can be rewritten as $\mathcal M(\mathcal A)\succ 0$, for all $\mathcal A$ such that  \eqref{Lem0_Proof_cond} holds true. Using an S-procedure, this statement is equivalent to the existence of a positive scalar $\eta>0$ 
such that
$$\begin{array}{l}
\begin{bmatrix}
\mathcal M_{1} &  \mathcal M_{2}
\mathcal A\\
\ast & \mathcal M_{3}
\end{bmatrix}\\-\eta 
\begin{bmatrix}
\begin{bmatrix} \boldsymbol{0}_{{n_1},{n_3}} \\ \mathcal Z^{\!\top}  \end{bmatrix} ^{\!\top}
\begin{bmatrix}
\mathcal R+ \mathcal N_1&\!\!\mathcal N_2\\
 \ast &\!\!\mathcal N_3\\
\end{bmatrix}^{-1} \begin{bmatrix} \boldsymbol{0}_{n_1,{n_3}} \\ \mathcal Z^{\!\top}  \end{bmatrix} 
^{\!\top} &  \mathcal Z
\mathcal A\\
\ast& \mathcal R
\end{bmatrix}\succ 0,
\end{array}
$$
which together with the Schur complement can be written
$$
\begin{bmatrix}
\mathcal M_{1} & (\mathcal M_{2}-\eta \mathcal Z)
\mathcal A & \boldsymbol{0}_{n_1,{n_3}} & \eta\mathcal Z  \\
\ast & \mathcal M_{3}-\eta \mathcal R& \boldsymbol{0}_{n_3}&\boldsymbol{0}_{{n_3},{n_2}}\\
\ast&\ast & 
\mathcal R+ \mathcal N_1&\!\!\mathcal N_2\\
 \ast &\ast &\ast &\!\!\mathcal N_3\\
\end{bmatrix}
\succ 0
$$
Selecting $\eta \mathcal R= \mathcal M_3-\epsilon I_{n_3}$, for any arbitrarily small $\epsilon>0$  and $\eta \mathcal Z =  \mathcal M_{2}$ so that the second (block) row and columns are zero and the uncertain matrix $\mathcal A$  disappears. The previous matrix becomes
$$
\begin{bmatrix}
\mathcal M_{1} & \boldsymbol{0}_{n,p} &\boldsymbol{0}_{n_1,{n_3}}&\mathcal M_2 \\
\ast & \epsilon I_{n_3}& \boldsymbol{0}_{n_3} &\boldsymbol{0}_{{n_3},{n_2}}\\
\ast&\ast &\mathcal M_{3}+\eta \mathcal N_1 -\epsilon I_{n_3}& \eta \mathcal N_2\\
\ast&\ast&\ast&\eta\mathcal N_3
\end{bmatrix}
\succ 0
$$
which is equivalent to \eqref{lem:IneqY} knowing that $\epsilon$ is an arbitrarily small scalar. 
\newline \underline{\textit{(ii)$\Rightarrow$(i):}} Pre- and post-multiply \eqref{lem:IneqY} by $\begin{bsmallmatrix}I_{n_1}&\boldsymbol{0}_{{n_1},{n_3}}& \boldsymbol{0}_{{n_1},{n_3}}\\ 
\boldsymbol{0}_{{n_3},{n_1}}& I_{n_3} & \mathcal A^{\!\top} \end{bsmallmatrix}$ and its transpose, respectively, leads to 
$$
\begin{array}{lcl}
0&\!\!\prec& \!\!\begin{bmatrix}
\mathcal M_{1} & \mathcal M_{2}\mathcal  A\\
\ast & \mathcal M_{3}
\end{bmatrix}+\eta \begin{bmatrix}
\boldsymbol{0}_{n_1}& \boldsymbol{0}_{n_1,{n_3}}\\
\ast &   \left[\begin{matrix}
I_{n_3}\\
\mathcal A\\
\end{matrix}\right]^{\!\!\top} \!\! \begin{bmatrix}
 \mathcal N_1&\!\!\mathcal N_2\\
 \ast &\!\!\mathcal N_3\\
\end{bmatrix} \left[\begin{matrix}
I_{n_3}\\
\mathcal A\\
\end{matrix}\right]
\end{bmatrix}\\
&\!\!\prec &\!\!\begin{bmatrix}
\mathcal M_{1} & \mathcal M_{2}\mathcal  A\\
\ast & \mathcal M_{3}
\end{bmatrix}
\end{array}
$$ 
where  the last inequality holds only for matrices $\mathcal A$ that belongs to $\Sigma_\mathcal N$, which concludes the proof. 
\end{proof}

\end{document}